\documentclass[conference,letterpaper]{IEEEtran}
\IEEEoverridecommandlockouts
\usepackage{setspace}
\usepackage[backend=biber,
style=ieee,sorting=none,doi=false]{biblatex}
\usepackage{amsmath,amssymb,amsfonts}
\usepackage{algorithm}
\usepackage{algpseudocode}
\usepackage{graphicx}
\usepackage{textcomp}
\usepackage{xcolor}
\def\BibTeX{{\rm B\kern-.05em{\sc i\kern-.025em b}\kern-.08em
    T\kern-.1667em\lower.7ex\hbox{E}\kern-.125emX}}

\usepackage[font=footnotesize]{caption}
\usepackage[T1]{fontenc}
\usepackage{amsfonts}
\usepackage{amssymb}
\usepackage{algorithm}
\usepackage{textcomp}
\usepackage{verbatim}
\usepackage{gensymb}
\usepackage{tabularray}

\usepackage{amsthm}
\usepackage{pgfplots}
\pgfplotsset{compat=1.17}
\usepgfplotslibrary{groupplots}
\usepackage{hyperref}
\hypersetup{
  colorlinks=false,
  linkbordercolor=white,
 urlbordercolor=white,
pdfborder={0 0 0}
}
\usepackage{multirow}
\usepackage{tabularx}
\usetikzlibrary{patterns}
\usetikzlibrary{arrows}
\usetikzlibrary{shapes.geometric}
\usetikzlibrary{calc} 
\usetikzlibrary{positioning}
\usetikzlibrary{arrows.meta}

\usepackage{subcaption}
\usepackage{tikz-3dplot}

\usetikzlibrary{pgfplots.groupplots}
\usepackage{xstring}
\newtheorem{proposition}{Proposition}[section]

\pgfdeclarepatternformonly[\LineSpace]{my north east lines}{\pgfqpoint{-1pt}{-1pt}}{\pgfqpoint{\LineSpace}{\LineSpace}}{\pgfqpoint{\LineSpace}{\LineSpace}}%
{
    \pgfsetlinewidth{0.4pt}
    \pgfpathmoveto{\pgfqpoint{0pt}{0pt}}
    \pgfpathlineto{\pgfqpoint{\LineSpace + 0.1pt}{\LineSpace + 0.1pt}}
    \pgfusepath{stroke}
}

\newdimen\LineSpace
\tikzset{
    line space/.code={\LineSpace=#1},
    line space=3pt
}

\tikzset{cross/.style={cross out, draw, 
         minimum size=2*(#1-\pgflinewidth), 
         inner sep=0pt, outer sep=0pt}}
\begin{document}

\title{Joint Transmit Signal and Beamforming Design for Integrated Sensing and Power Transfer Systems
\thanks{This work was (partly) funded by the Deutsche Forschungsgemeinschaft (DFG, German Research Foundation) – SFB 1483 – Project-ID 442419336, EmpkinS.}
}

% \author{\IEEEauthorblockN{1\textsuperscript{st} Kenneth MacSporran Mayer}
% \IEEEauthorblockA{\textit{Institute Digital Communications} \\
% \textit{Friedrich-Alexander-University}\\
% Erlangen, Germany \\
% kenneth.m.mayer@fau.de}
% \and
% \IEEEauthorblockN{2\textsuperscript{nd} Laura Cottatellucci}
% \IEEEauthorblockA{\textit{Institute Digital Communications} \\
% \textit{Friedrich-Alexander-University}\\
% Erlangen, Germany \\
% laura.cottatellucci@fau.de}
% \and
% \IEEEauthorblockN{3\textsuperscript{rd} Robert Schober}
% \IEEEauthorblockA{\textit{Institute Digital Communications} \\
% \textit{Friedrich-Alexander-University}\\
% Erlangen, Germany \\
% robert.schober@fau.de}
% }
% \vspace*{-0.4cm}
% }
\author{%\vspace*{-0.05cm}
\IEEEauthorblockN{Kenneth MacSporran Mayer, Nikita Shanin, Zhenlong You, Sebastian Lotter, Stefan Brückner, Martin Vossiek,\\ Laura Cottatellucci, and Robert Schober}
\IEEEauthorblockA{Friedrich-Alexander-Universität Erlangen-Nürnberg, Germany}%\vspace*{-3cm}
}

\maketitle
%\vspace*{-0.2cm}
\begin{abstract}
Integrating different functionalities, conventionally implemented as dedicated systems, into a single platform allows utilising the available resources more efficiently.
We consider an integrated sensing and power transfer (ISAPT) system and propose the joint optimisation of the rectangular pulse-shaped transmit signal and the beamforming vector to combine sensing and wireless power transfer (WPT) functionalities efficiently.
In contrast to prior works, we adopt an accurate non-linear circuit-based energy harvesting (EH) model.
We formulate and solve a non-convex optimisation problem for a general number of EH receivers to maximise a weighted sum of the average harvested powers at the EH receivers while ensuring the received echo signal reflected by a sensing target (ST) has sufficient power for estimating the range to the ST with a prescribed accuracy within the considered coverage region.
The average harvested power is shown to monotonically increase with the pulse duration when the average transmit power budget is sufficiently large.
We discuss the trade-off between sensing performance and power transfer for the considered ISAPT system. 
The proposed approach significantly outperforms a heuristic baseline scheme based on a linear EH model, which linearly combines energy beamforming with the beamsteering vector in the direction to the ST as its transmit strategy.  
\end{abstract}

\section{Introduction}\label{Section: Introduction}
The growing number of Internet-of-Things (IoT) applications, ranging from smart homes to healthcare, requires the deployment of a multitude of low-power IoT devices, which future wireless networks will have to serve by facilitating communication, sensing, computation, and a supply of power \cite{Saad20,Tong21}.
An efficient method for powering these low-power IoT devices through radio-frequency (RF) signals is known as wireless power transfer (WPT), which may even allow for battery-free operation of the devices \cite{Clerckx19}.
In light of the tremendous number of devices, the efficient utilisation of resources, such as energy, frequency spectrum, and hardware, is indispensable.
Moreover, powering the devices may require localising them first.
This especially applies to the crowded sub-6 GHz band, which is occupied by different systems
 \cite{Liu22-selected,Kaushik22}.
An effective method of utilising the scarce resources more efficiently is the integration of different functionalities, such as communication, sensing, and WPT, by co-designing them into one platform, instead of employing dedicated systems for each functionality.
%, which share the same bandwidth.
Currently, the most prominent example of this approach is integrated sensing and communications (ISAC), which has received significant attention and is envisioned as a key technology for next-generation wireless systems \cite{Liu22}.
Hereby, sensing may be employed to acquire information on a target's location by considering the time delay between an emitted signal and the resulting echo signal reflected by the target \cite{Liu22, Skolnik81}.

%Unlike ISAC, the integration of other functionalities, such as integrated sensing and power transfer (ISAPT), has received significantly less interest.
Unlike ISAC, the integration of sensing and power transfer (ISAPT) has received significantly less attention, despite facilitating the decongestion of spectrum and the design of hardware-efficient systems by reducing the size, cost, and power consumption.
%Nevertheless, ISAPT systems facilitate the decongestion of spectrum and the design of hardware-efficient systems by reducing the size, cost, and power consumption  of the overall system \cite{Yang23,Ping22,Li23,Chen23,Rezaei23}.
%Applications of ISAPT system include, e.g., future transportation systems \cite{Li23}. 
Only a few recent studies are available on ISAPT \cite{Yang23,Ping22,Li23,Chen23,Rezaei23}. 
In \cite{Yang23, Ping22}, the trade-off between WPT and sensing is investigated by optimising the transmit beamforming vector, whereby the authors of \cite{Ping22} assumed the energy harvesting (EH) receivers to be in the radiating near-field of the transmit antenna.
%Furthermore, in \cite{Li23arxiv}, an ISAPT system is assisted by an intelligent reflecting surface (IRS) and the beamforming design and the IRS phase shifts are jointly optimised. 
Moreover, in \cite{Li23, Chen23, Rezaei23}, triple-functionality systems are considered, integrating communication, WPT, and sensing into one system by either considering all three functionalities simultaneously \cite{Li23,Chen23}, or employing ISAPT in the downlink and transmitting data in the uplink \cite{Rezaei23}.
%In \cite{Li23}, the beamforming design is optimised to achieve communication and WPT targets while improving the sensing performance.
%In \cite{Chen23}, the trade-off between communication, sensing, and WPT is investigated.
%by optimising the covariance matrix of the Gaussian transmit signal, which is the optimal waveform for communication.
%In \cite{Rezaei23}, the sensing and communication performance is optimised by designing radar receive filters and the radar transmit signal.
%unmanned aerial vehicle trajectories, time scheduling, and uplink powers.

The WPT design in \cite{Yang23,Ping22,Li23,Chen23,Rezaei23} is based on a linear EH model.
Thus, these works do not capture the well-documented non-linear behaviour of practical EH circuits \cite{Clerckx19, Kim20}.
However, accurately taking this non-linear behaviour into account, e.g., by employing the circuit-based EH model in \cite{Morsi20}, is vital when the maximisation of \textit{harvested} power instead of \textit{received} power is desired \cite{Clerckx19, Kim20, Morsi20, Shanin22}.
%Regarding the sensing component of the ISAPT system, existing works \cite{Yang23,Chen23,Ping22,Li23,Rezaei23,Li23arxiv} characterise radar sensing performance by an estimation of the angle, i.e., the direction, towards the sensing target. 
%To this end, the beamforming design is matched to the desired beampattern \cite{Yang23,Ping22,Li23arxiv}.
%Alternatively, the direction is estimated by minimising the Cramer-Rao Bound of the direction \cite{Li23,Chen23}. 
%We \cite{Skolnik81}
In \cite{Yang23,Ping22,Li23,Chen23,Rezaei23}, the focus lies on optimising the beamforming design of the respective ISAPT systems while employing a transmit signal, which is favourable for one of the individual functionalities.
Specifically, a continuous wave pulse signal is utilised for facilitating both sensing and WPT in \cite{Rezaei23}, whereas random unit-variance signals are exploited in \cite{Yang23,Ping22,Li23,Chen23}.
%While EH based on a signal designed for radar sensing is possible,
Although EH receivers can opportunistically harvest power from a signal designed for radar sensing, this approach may not fully exploit the ISAPT system's potential, whereas jointly co-designing the transmit signal for both functionalities can offer performance benefits.
%, whereas a transmit signal merely tailored to one of the system's functionalities may not fully exploit the available potential of the system.
Indeed, in simultaneous wireless information and power transfer systems the optimal transmit signal represents a trade-off between the signals optimal for both communication and WPT \cite{Morsi20}.
%Consequently, employing a transmit signal merely tailored to one of the system's functionalities, may limit the performance. 
For example, while the random Gaussian transmit signal employed in \cite{Yang23,Ping22,Li23,Chen23} is optimal for communication, it is highly suboptimal for WPT \cite{Morsi20}.
Moreover, the utilisation of random transmit signals for WPT and sensing prevents the prediction of the harvested power and the sensing performance in a given time slot, respectively.
In fact, for WPT, the transmission of rectangular pulse-shaped signals is typically assumed \cite{Morsi20,Shanin22}.
Moreover, the optimal transmit strategy for WPT has been investigated in \cite{Morsi20,Shanin22}, where on-off signaling was found to be optimal for single-user WPT systems.
This signaling strategy is inherently similar to a pulse time-delay radar transmitting a rectangular pulse during the on-period and subsequently awaiting the reception of the echo signal reflected by the target during the off-period \cite{Skolnik81}.
Therefore, in this paper, we propose to leverage this similarity to facilitate the efficient integration of WPT and sensing.
%Moreover, a pulse time-delay radar allows for low complexity echo signal processing by employing the leading edge technique \cite{Skolnik81}.

%Information regarding a sensing target's (ST) location is attained by comparing the emitted signal of the sensing system to the echo signal reflected by the target. 
%Thus, adequate sensing performance can only be ensured if sufficient time for receiving the echo signal is provided, which is not considered in existing works \cite{Yang23,Ping22,Li23,Li23arxiv,Rezaei23,Chen23}. 

In this paper, we consider an ISAPT system comprising a multi-antenna transceiver (TRX), a single ST, and multiple EH receivers.
The considered system employs a round-trip time-of-flight pulse time-delay radar for sensing.
While the majority of the existing works focus on determining the angular direction of the ST with respect to the TRX \cite{Yang23,Ping22,Li23,Chen23}, the proposed ISAPT system is designed to ensure the range to the ST in a certain direction is estimated with a desired accuracy by considering the time delay between the emitted transmit signal and the received echo signal.
Our goal is the joint optimisation of the transmit signal and the transmit beamforming vector for ISAPT systems for integrating the sensing and WPT functionalities.
%instead of relying on a transmit signal designed for a single functionality.
%efficient for both WPT and sensing and the transmit beamformer in ISAPT systems 
To this end, we formulate an optimisation problem to maximise the weighted sum of the average harvested powers at the EH receivers while ensuring a desired sensing performance and optimise the amplitude and duration of the rectangular transmit pulse in conjunction with the transmit beamforming vector.
In contrast to \cite{Yang23,Ping22,Li23,Chen23,Rezaei23}, we adopt the non-linear circuit-based EH model derived in \cite{Morsi20}, which more accurately characterises the power harvested in practical EH circuits compared to linear EH models \cite{Clerckx19, Kim20}.
We characterise the sensing performance by ensuring the minimum power of the echo signals received from the ST is sufficient for estimation of the range to the ST with a desired accuracy within the considered coverage region.
Guaranteeing a particular accuracy for estimation of the range to the ST and the adopted non-linear EH model are key differences to existing works.
%, whereby sufficient time for receiving the echo signal reflected of the ST is provided to ensure the desired sensing performance.
%Moreover, a minimum distance to the ST is enforced, which limits the duration of the transmitted pulse.
%By solving the proposed problem, we obtain the optimal pulse duration, the signal amplitude, and the beamforming vector.
%The non-convex optimisation problem is formulated for a general number of EH receivers and a single ST with the objective of maximising the sum of harvested energies and we solve the problem for a single EH receiver and the single ST.
Prior to the solution of the proposed non-convex optimisation problem, the feasibility region of the pulse duration is determined analytically.
%by taking into account the constraints on range accuracy and minimum distance imposed by the sensing system and the coherence time of the channels to the EHs. 
We then propose a solution based on semidefinite relaxation (SDR) and successive convex approximation (SCA).
%and we prove that the solution of the relaxed problem satisfies the constraints of the original problem for a given pulse duration analytically. 
The pulse duration yielding the largest amount of average harvested power at the EH receivers is determined through a grid search.
Moreover, we discuss the trade-off between sensing performance and WPT and show that the proposed solution significantly outperforms a heuristic baseline scheme, which is based on a combination of the optimal radar and the optimal WPT transmit signals when assuming a linear EH model.
Hence, the results presented in this paper deliver novel insights for the efficient integration of sensing and WPT functionalities into one common platform.

\section{System Model}\label{Section: System Model}
In this paper, we consider an ISAPT system that comprises a TRX employing a uniform linear array (ULA) equipped with $N_\mathrm{t} \geq 1$ antennas, $M \geq 1$ single-antenna EH receivers, and a single ST.
In Section \ref{Section: Transmit Signal Model}, the transmit signal model is discussed and the ISAPT system's WPT and sensing functionalities are presented in Section \ref{Section: WPT System} and Section \ref{Section: Sensing System}, respectively.
%We consider line-of-sight (LoS) conditions to the ST and a Rician fading channel between the transmitter and the EH receiver.

\subsection{Transmit Signal Model}\label{Section: Transmit Signal Model}
The TRX emits a pulse-modulated RF signal with the following equivalent complex baseband (ECB) representation
\begin{align}\label{eq: transmit signal}
    \boldsymbol{x}(t) = \sum_{k^\prime = -\infty}^\infty \boldsymbol{x}[k^\prime] \psi (t-k^\prime T;\tau[k^\prime]) = \boldsymbol{x}[k] \psi (t-kT;\tau[k]),
\end{align}
where $k= \text{max}\left\{ y \in \mathbb{Z} \vert y \leq t/T \right\}$ with $\mathbb{Z}$ denoting the set of integers, $\boldsymbol{x}[k] \in \mathbb{C}^{N_\mathrm{t} \times 1}$ is the transmit signal vector in the $k$-th time slot, and $\psi(t;\tau[k])$ represents a rectangular pulse with magnitude $1$ and duration $\tau[k]$, i.e., $\psi(t;\tau[k])=1$ if $0 \leq t \leq \tau[k]$, and zero otherwise.
%is defined as 
% \begin{align}\label{eq: transmit pulse}
%     \psi(t;\tau) =
%     \begin{cases}
% 	    1, \quad 0 \leq t \leq \tau \\
%      0, \quad \text{otherwise},
%     \end{cases}
% \end{align}
%which .
Furthermore, the duration of a time slot is denoted by $T > \tau[k]$.
%For WPT, the transmission of rectangular pulse-shaped signals is commonly assumed and an on-off signaling strategy was shown to be optimal, at least for single-user WPT systems \cite{Morsi20,Shanin22}.
%We note the similarity between this transmit signaling strategy and pulse time-delay radar, which also emits rectangular pulses and subsequently, awaits the reception of the echo signal reflected by the target \cite{Skolnik81}.
We note that pulse-modulated signals as in \eqref{eq: transmit signal} are commonly utilised for the design of WPT systems \cite{Morsi20,Shanin22} and pulse radar systems \cite{Skolnik81}.
%Recall from Section \ref{Section: Introduction}, the similarity between on-off signaling in WPT, at least for multi-user SISO and single-user MISO WPT systems \cite{Morsi20,Shanin22}, and pulse time-delay radar.
Leveraging this similarity offers a seamless integration of WPT and sensing functionalities, and thus, we exploit the pulse time-delay radar concept in the proposed ISAPT system.
%Moreover, a pulse time-delay radar allows for low complexity echo signal processing by employing the leading edge technique \cite{Skolnik81}.
%A pulse time-delay radar employing a rectangular pulse shape is employed in the proposed ISAPT system since the transmit signaling strategy comprising an on-period when transmitting the pulse and an off-period while awaiting the echo reflected by the ST is inherently similar to the optimal transmit strategy for a multiple-input single-output WPT system, which is employs a rectangular pulse shape and on-off signaling \cite{Shanin22}.
%Moreover, this allows for low complexity echo signal processing by employing the leading edge technique \cite{Skolnik81}.
%By employing a pulse time-delay radar as the sensing system, an efficient integration of the WPT and the sensing functionalities is achieved by leveraging the inherent similarity of the rectangular waveforms.
We set, %$\boldsymbol{x}[k]$ is given by
%\begin{align}\label{eq: TX signal rectangle}
    $\boldsymbol{x}[k] = A[k] \boldsymbol{w}[k]$,
%\end{align}
where $A[k] \in \mathbb{R}$ is the signal amplitude and $\boldsymbol{w}[k] \in \mathbb{C}^{N_\mathrm{t} \times 1}$ is the beamforming vector.
The energy of $\boldsymbol{w}[k]$ is normalised to unity, i.e., $\| \boldsymbol{w}[k] \|_2^2 = 1$, where $\| \cdot \|_2$ is the Euclidean norm.
%In summary, the transmit signal in time slot $k$ is a rectangular pulse with duration $\tau$ and amplitude $A$ which is steered in the direction of $\boldsymbol{w}$, where the dependency on $k$ is dropped for better legibility.
Thus, $\tau[k]$, $A[k]$, and $\boldsymbol{w}[k]$ represent the optimisation variables for jointly optimising the duration, amplitude, and beamforming of the transmit signal, respectively.
%the rectangular pulse shape and the beamforming design.
%In the following, we focus on a single time slot for the joint transmit signal and beamforming design and thus, 
The dependence of $\tau[k]$, $A[k]$, and $\boldsymbol{w}[k]$ on symbol interval $k$ is dropped in the notation in the remainder of this paper for better legibility.

\subsection{Wireless Power Transfer System}\label{Section: WPT System}
We assume a fading channel $\boldsymbol{h}_m \in \mathbb{C}^{N_\mathrm{t} \times 1}$ between the TRX and EH node $m$ with $m=1,\dots,M$. %with a channel coherence time of $T_\mathrm{coh}$ \cite{Shanin22}.
Moreover, perfect knowledge of the channel vectors $\boldsymbol{h}_m$, $\forall m=1,\dots,M$, is assumed at the TRX.
Consequently, the instantaneous power of the received signal at the $m$-th EH receiver in interval $[kT, kT + \tau]$ is given by
\begin{align}\label{eq: EH received power}
    P_m = P_m(A\boldsymbol{w}) = A^2 \vert \boldsymbol{h}_m^H \boldsymbol{w} \vert^2.
\end{align}
For WPT, we neglect the impact of additive noise due to its negligible contribution to the average harvested power \cite{Shanin22}.

Each EH receiver is equipped with a rectenna comprising an antenna, a matching circuit, a non-linear rectifier, such as a Schottky diode combined with a low-pass filter, and a load resistor \cite{Morsi20,Shanin22}.
We adopt the non-linear circuit-based EH model proposed in \cite{Morsi20}, which was derived by accurately analysing the current flow through the electrical EH circuit \cite{Morsi20}. 
The harvested power $\Tilde{\varphi}(P_m)$ from the pulse at the $m$-th EH receiver in time slot $k$ is given by 
%\begin{align}\label{eq: harvested power m-th EH}
$\Tilde{\varphi}(P_m) = \mathrm{min} \left\{ \varphi(P_m), \varphi(P_{\mathrm{max}}) \right\}$,
%\end{align}
which is bounded due to the saturation of practical EH circuits caused by the breakdown of the employed Schottky diodes at high received powers \cite{Morsi20}.
%is valid for input powers up to $P_{\mathrm{max}}$, i.e., $P_m \leq P_{\mathrm{max}}$, as the rectenna circuit is driven into saturation for input powers $P_m > P_{\mathrm{max}}$ \cite{Morsi20}.
%and define the harvested power of the rectenna as a function of the power of the received ECB signal \eqref{eq: EH received power} \cite{Shanin22}.
Hereby, the non-linear monotonically increasing function $\varphi(P_m)$ models the EH circuit and is given by \cite{Morsi20,Shanin22}
\begin{align}\label{eq: EH non-linearity}
    \varphi(P_m) = \left[ \frac{1}{a} \mathrm{W}_0\left( a \mathrm{e}^a \mathrm{I}_0\left( C \sqrt{2 P_m} \right) \right) -1 \right]^2 I_\mathrm{s}^2 R_\mathrm{L},
\end{align}
where $\mathrm{W}_0(\cdot)$ and $\mathrm{I}_0(\cdot)$ are the principal branch of the Lambert-W function and the zeroth order modified Bessel function of the first kind, respectively. 
The parameters $a$ and $C$ depend on the employed EH circuit and are independent of the received signal \cite{Morsi20,Shanin22}. 
Moreover, $I_\mathrm{s}$ is the reverse bias saturation current of the diode and $R_\mathrm{L}$ is the resistance of the load.
Note that operating Schottky diodes in the breakdown regime should be avoided \cite[Remark 5]{Clerckx18}.
Therefore, we enforce $P_m \leq P_{\mathrm{max}}$ and consequently, $\Tilde{\varphi}(P_m) = \varphi(P_m)$.
%The parameters of the rectenna circuit $I_\mathrm{s}$, $R_\mathrm{s}$, $\mu$, $V_\mathrm{T}$, and $Z_\mathrm{a}^*$ denote the reverse bias saturation current, the series impedance, the ideality factor of the diode, the thermal voltage, and the complex-conjugate of the input impedance, respectively.
%The load impedance is defined as $R_\mathrm{L}$.
%Note that the rectenna circuit is driven into saturation for large input powers $P_{\mathrm{max}}$ and thus, $\varphi(P_m)$ is bounded by $\varphi(P_{\mathrm{max}})$. 
%Considering this saturation effect, 
The average harvested power at the $m$-th EH node in time slot $k$ is determined by averaging the power received during pulse duration $\tau$ over time slot $T$.
We consider the weighted sum of average harvested powers among all EH receivers during time slot $k$, which is given by 
\begin{align}\label{eq: sum harvested power}
    \phi(A\boldsymbol{w}, \tau) = \frac{\tau}{T} \sum_{m=1}^M \beta_m \varphi(P_m),
\end{align}
where $\beta_m \in [0,1]$, $\forall m=1,\dots,M$, is the weight associated with the $m$-th EH receiver such that $\sum_m \beta_m = 1$.
The $\beta_m$, $\forall m=1,\dots,M$, can be chosen to, e.g., ensure EH fairness among the EH nodes.

\subsection{Sensing System}\label{Section: Sensing System}
We define the location of an ST by the tuple $\left(R,\alpha \right)$, whereby $R$ is the range from the TRX to the ST and $\alpha$ is the angular direction of the ST.
%While the majority of existing works \cite{Yang23,Ping22,Li23,Chen23} determine whether a target is present in the direction of the angular direction $\alpha$ to the ST, 
The objective of the proposed ISAPT system is to ensure the range $R$ to the ST at angular direction $\alpha$ can be estimated to a desired accuracy.
Specifically, angular direction $\alpha$ is probed within a certain time frame $T_\mathrm{sen}$, whereby a meaningful echo signal is only received if a ST is present in direction $\alpha$.
The absence of a meaningful echo signal implies no ST is present at $\alpha$ and thus, a different angular direction $\alpha^{\prime}$ is explored.
%Given the ST is located at angular direction $\alpha$, we propose an ISAPT system for estimating the range $R$ to the ST, thereby localising the ST.
Here, we focus on the case when a meaningful echo signal is detected from angular direction $\alpha$.
%, where the estimation of $\alpha$ may be implemented through, for example, the Cramer-Rao Bound as presented in \cite{Li23}. 
%a method presented in existing works \cite{Yang23,Ping22,Li23,Chen23}.
%The minimum duration of this time frame is the time slot duration $T$, which is determined in Section \ref{Section: bounds coverage range} to ensure reliable sensing performance.
%The duration of this time frame 
The duration of time frame $T_\mathrm{sen}$ is assumed fixed and determined by the characteristics of the TRX and upper-bounded by the coherence time $T_\mathrm{coh}$ of the fading channels between the TRX and the EH receivers, i.e., $T \leq T_\mathrm{sen} \leq T_\mathrm{coh}$.
%since the ST may be co-located with an EH node in a single device. 
In this paper, we assume $T_\mathrm{coh}$ and $T_\mathrm{sen}$ are on the order of milliseconds, whereas time slot duration $T$ is typically on the order microseconds to nanoseconds \cite{Skolnik81}, i.e., $T_\mathrm{sen} \gg T$.
The location of the ST is assumed to be quasi-static, i.e., the ST remains approximately static during $T_\mathrm{sen}$.
%, and thus, we set the sensing time frame to $T_\mathrm{sen} = T_\mathrm{coh}$.
During $T_\mathrm{sen}$, the ISAPT system uses pulses from multiple time slots, each of duration $T$, to perform coherent pulse integration, which is discussed in detail in Section \ref{Section: Sensing accuracy}.

Next, we establish the relationship between the coverage range, within which an estimation of the range $R$ is desired, the pulse duration $\tau$, and the time slot duration $T$.

% A-priori information regarding the approximate range $\Tilde{R}$ to the ST, which may be available from a previous time slot, allows for a more efficient allocation of energy and time for simultaneous ST localisation and WPT.
% In the following, we assume this information is available\footnote{The impact of the reliability of this prior information, i.e., the difference between the actual and the approximate distances to the ST, $R$ and $\Tilde{R}$, respectively, on the system performance presents an interesting direction for future work, but is omitted in this paper due to space limitations.}.
% An estimation of the range $R$ in the absence of this information is more resource-demanding since any feasible range to the ST must be considered including the maximum range $R_\mathrm{max}$ spanning the boundary of the coverage region.
% Consequently, more signal power than necessary for estimating the range $R$ may be required, when the ST is located at $R < R_\mathrm{max}$.
%We focus on the estimation of the range $R$ to the ST while assuming knowledge of the angular direction $\alpha$, which can be estimated through methods presented in existing works, such as \cite{Yang23,Ping22,Li23,Li23arxiv,Chen23}.

\subsubsection{Bounds on $\tau$ and $T$}\label{Section: bounds coverage range}
While transmitting, the TRX cannot receive a reflected echo signal, which is necessary for localising the ST \cite{Skolnik81}.
Therefore, an upper bound on pulse duration $\tau$ is necessary.
This is determined by the minimum range $R_\mathrm{min}$ the ST may be located at \cite{Skolnik81}, i.e.,
\begin{align}\label{eq: minimum distance constraint}
    \tau \leq \tau_\mathrm{max} = \frac{2\, R_\mathrm{min}}{c},
\end{align}
where $c$ is the speed of light.
%Hereby, .
%Note that \eqref{eq: minimum distance constraint} is mathematically equivalent to the definition of the range resolution of a radar sensing system, i.e., the capability of the system to resolve two targets.
%Besides quantifying a radar system's performance through the accuracy of the range measurement, additional performance metrics are available such as the range resolution $R_R$.
The duration of a time slot $T$ is constrained by the maximum coverage range $R_\mathrm{max}$.
Specifically, the shortest possible time slot duration $T_\mathrm{min}$ to ensure a target is reached anywhere within the coverage region and the echo signal is received back when employing a very short rectangular pulse, i.e., $\tau \to 0$, is
%\begin{align}\label{eq: minimum time frame duration}
    $T \geq T_\mathrm{min} = (2 \, R_\mathrm{max})/c$ \cite{Skolnik81}.
%\end{align}
%The range $R$ to the ST is bounded by the coverage range.
Moreover, since $R_\mathrm{min} \leq R \leq R_\mathrm{max}$, time slot duration $T$ can be set as follows
\begin{align}\label{eq: time slot}
    T(\tau) = T_\mathrm{min} + \tau = \frac{2 \, R_\mathrm{max}}{c} + \tau, 
\end{align}
where $T(\tau)$ exceeds the minimum required time $T_\mathrm{min}$ by the pulse duration $\tau$, which is yet to be optimised for ISAPT, such that the full echo signal is received before the next pulse is transmitted.
%Therefore, time slot duration $T(\tau)$ is just enough time to receive . reaching a target anywhere within the coverage region and receiving the echo back .
%Hereby, the transmit pulse duration is assumed to be identical to the duration of the echo pulse \cite{Skolnik81}.

Next, we characterise the accuracy of estimating the range $R$ to guarantee the desired sensing performance within the coverage region $R_\mathrm{min} \leq R \leq R_\mathrm{max}$.

\subsubsection{Sensing accuracy}\label{Section: Sensing accuracy}
%A particular accuracy of estimating the range in the full coverage region can be guaranteed 
By ensuring that the minimum echo signal power $P_\mathrm{ST}$, i.e., the power of the echo received from the ST at $R_\mathrm{max}$, is sufficient for attaining a particular accuracy of estimating the range, the desired accuracy is guaranteed to be attained in the full coverage region, i.e., $R \leq R_\mathrm{max}$.
%The required accuracy of estimating the range within the coverage region is achieved by ensuring the minimum echo signal power, i.e., the power of the echo received from the ST at $R_\mathrm{max}$, is sufficient for attaining the necessary accuracy.
The minimum echo signal power $P_\mathrm{ST}$ during a time slot follows from the radar equation and is given by \cite{Skolnik81}
\begin{align}\label{eq: min echo power}
    \frac{\tau}{T(\tau)} P_\mathrm{ST} = \frac{\tau \, A^2}{T(\tau)} \vert \boldsymbol{u}^H \boldsymbol{w} \vert^2 \underbrace{\frac{\lambda^2 \sigma_\mathrm{RCS}}{(4 \pi)^3 R_\mathrm{max}^4} \| \boldsymbol{u} \|_2^2}_{=z_1},
\end{align}
where $\boldsymbol{u} = \left[1,\mathrm{e}^{q}, \mathrm{e}^{2q},\dots, \mathrm{e}^{(N_\mathrm{t}-1)q}\right]^T \in \mathbb{C}^{N_\mathrm{t} \times 1}$ with $q=-\mathrm{j}\pi\mathrm{sin}(\alpha)\Delta_\mathrm{TRX}$ denotes the vector of phase delays in direction $\alpha$, $\Delta_\mathrm{TRX}$ is the transmit antenna spacing, $\lambda$ is the wavelength, and $\sigma_\mathrm{RCS}$ is the radar cross-section (RCS) of the ST.
The receive beamformer is chosen according to direction $\alpha$ for maximising the signal power at the detector.
%Note that the elements of the ULA are spaced by $\lambda/2$ \cite{massivemimobook}.

The received echo at the TRX\footnote{The study of additional disturbing influences, such as additional system and propagation losses, on the minimum echo signal power and thus, the system performance, is an interesting topic for future work. 
Here, we aim at determining the maximum achievable system performance, and therefore, we neglect these additional influences.} is impaired by additive white Gaussian noise (AWGN) with noise power $\sigma_\mathrm{n}^2$. 
In this paper, we consider range estimation based on the estimated time delay $T_R$ between the emission of the transmit signal and the reception of the echo signal \cite{Skolnik81}.
%Hereby, the estimated range $R_\mathrm{est}$ is defined as 
%\begin{align}
%    R_\mathrm{est} = T_R \, \frac{c}{2},
%\end{align}
%where $T_R$ is the estimate of the true time delay $T_0$. 
%In a practical, non-idealistic radar system the estimation of the range $R$ to a target is erroneous.
Hence, the sensing performance of the ISAPT system is characterised by the accuracy with which the system is capable of determining time delay $T_R$ \cite{Skolnik81}.
Since the estimated range is a linear function of $T_R$, the estimation error of the range, $\hat{R}$, defined as the root mean squared (RMS) error between the estimated range and the true range, is given by $\hat{R} = (\hat{T}_R\,c/2)$, where $\hat{T}_R$ is the RMS error of the estimated time delay $T_R$ \cite{Skolnik81}.
Hereby, the accuracy of estimating the time delay is limited by noise. 
Thus, a large signal-to-noise ratio (SNR) is required to obtain an accurate estimate of $T_R$ and thus, typically, the echoes of multiple pulses are exploited for the estimation of $T_R$ \cite{Skolnik81}.
To this end, the proposed ISAPT system employs coherent pulse integration.
Consequently, the total signal power at the TRX in time frame $T_\mathrm{sen}$ grows linearly with the number of pulses per $T_\mathrm{sen}$.
For maximum performance, $N(\tau) = \text{max}\left\{ y \in \mathbb{Z} \vert y \leq T_\mathrm{sen} / T(\tau) \right\}$ pulses are employed for estimating $T_R$ since the echo signals received in different time slots are added coherently, whereas the noise in different time slots is assumed to be uncorrelated \cite{Skolnik81}.
In the following, we assume $N(\tau) = T_\mathrm{sen} / T(\tau)$ for analytical tractability.
%As a consequence of the increased SNR, the impact of false target detection is neglected.
Moreover, the impact of multipath components is neglected since they arrive after the line-of-sight (LoS) echo signal.
%, in contrast to the , the combination of the multipath components would arrive after the LoS echo.
As a result, for the considered rectangular pulse shape, $\hat{T}_R$ is given by \cite{Skolnik81}
\begin{align}\label{eq: time-delay rms}
    \hat{T}_R = \frac{1}{B} \sqrt{\underbrace{\frac{\sigma_\mathrm{n}^2}{4 \, T_\mathrm{sen}}}_{=z_2} \frac{[T(\tau)]^2}{\tau \, P_\mathrm{ST}}},
\end{align}
where $B$ is the bandwidth available to the ISAPT system.
%where $t_\mathrm{r}$ is the rise time of the pulse and is approximately the inverse of the bandwidth $B$, i.e., $t_\mathrm{r} \approx 1/B$.
%Note that the approximation in \eqref{eq: time-delay rms} is tight for an ideal rectangular pulse with infinite bandwidth \cite{Skolnik81}.
%Expression \eqref{eq: time-delay rms} can be derived using measurement methods such as the leading-edge technique or optimum processing techniques employing gating signals and matched filters \cite{Skolnik81}.

\section{Problem Formulation and Proposed Solution}\label{Section: Problem Formulation}
\subsection{Problem Formulation}
The objective is to maximise the weighted sum of the average harvested powers at the EH receivers \eqref{eq: sum harvested power} while ensuring the desired accuracy of estimating the range $R$ by limiting the RMS error $\hat{R}$ to lie below a certain threshold $\hat{R}_\mathrm{max}$, which defines the maximum tolerated sensing error.
To this end, we jointly optimise the transmit beamforming and the duration and amplitude of the rectangular transmit pulse in symbol interval $k$.
Mathematically, this is formulated as the following optimisation problem
\begin{subequations}\label{P1: Original Problem}
\begin{alignat}{2}
&\underset{\tau, A, \boldsymbol{w}}{\text{maximise}}
&\qquad& \! \phi(A\boldsymbol{w},\tau)
\label{Objective1: Sum harvested power}\\
&\text{subject to} 
&& \text{C1: } z \, \sqrt{\frac{[T(\tau)]^2}{\tau \, A^2 \vert \boldsymbol{u}^H \boldsymbol{w} \vert^2}} \leq \hat{R}_\mathrm{max}, \label{con1: radar range accuary}\\
%\displaybreak
&&& \text{C2: } \frac{\tau}{T(\tau)} \, A^2 \| \boldsymbol{w} \|_2^2 \leq P_\mathrm{avg}, \label{con1: average transmit power} \\
&&& \text{C3: } A^2 \| \boldsymbol{w} \|_2^2 \leq P_\mathrm{p}, \label{con1: peak transmit power} \\
&&& \text{C4: } A^2 \vert \boldsymbol{h}_m^H \boldsymbol{w} \vert^2 \leq P_\mathrm{max}, \forall m=1\dots M, \label{con1: peak input power}  \\
&&& \text{C5: } 0 \leq \tau \leq \tau_\mathrm{max} \label{con1: tau feasibility}
\end{alignat}
\end{subequations}
where $z=(c \sqrt{z_2})/(2\, B\sqrt{z_1}) $ with $z_1$ and $z_2$ defined in \eqref{eq: min echo power} and \eqref{eq: time-delay rms}, respectively.
Moreover, we impose constraints on the average transmit power $P_\mathrm{avg}$ and the peak transmit power $P_\mathrm{p}$ in \eqref{con1: average transmit power} and \eqref{con1: peak transmit power}, respectively.
Note that \eqref{con1: peak input power} avoids the EH receivers operating in the breakdown regime by limiting the peak received power to $P_\mathrm{max}$.
Problem \eqref{P1: Original Problem} is non-convex due to the objective function \eqref{Objective1: Sum harvested power} and constraints C1, C2, C3, and C4.
%The duration of a time slot $T$ is treated as a system parameter and is defined in Section \ref{Section: Feasibility region}.  

%\section{Proposed Solution}\label{Section: Optimal Solution}
\subsection{Feasibility Region of Problem \eqref{P1: Original Problem}}\label{Section: Feasibility region}
Prior to solving Problem \eqref{P1: Original Problem}, we determine the feasibility region $\mathcal{T}$ of $\tau$ analytically by taking into account the constraint on the sensing accuracy C1 and the upper bound imposed in C5.
\begin{proposition}\label{theorem: Feasibility region of pulse duration}
The pulse duration $\tau$ has the feasible region $\mathcal{T} = \left[ \tau_\mathrm{min}, \tau_\mathrm{max} \right]$, where the minimum pulse duration $\tau_\mathrm{min}$ is given by
\begin{align}
    \tau_\mathrm{min} = \frac{1}{2} \left( z_3 - z_4 - \sqrt{z_3^2 - 2z_3 z_4} \right)\mathrm{,}
\end{align}
with $z_3 =P_\mathrm{p} \| \boldsymbol{u} \|_2^2 \hat{R}_\mathrm{max}^2 / z^2 > 0$ and $z_4 = 4R_\mathrm{max}/c > 0$.
\end{proposition}
\begin{proof}
    The proof is provided in Appendix \ref{appendix A}.
\end{proof}

%\begin{corollary}\label{corollary: C1}
%    The lower bound on the achievable accuracy level of the ISAPT system, which follows from C1, is given by
%    \begin{align}\label{eq: lowest acc}
%        \hat{R}_\mathrm{max} \geq \sqrt{\frac{R_\mathrm{max}^5 \, c \, \sigma_n^2 \,(4\pi)^3}{2 \, P_\mathrm{p} \| \boldsymbol{u} \|_2^4 \, B^2 \, T_\mathrm{coh}\, \lambda^2 \, \sigma_{RCS} }}.
%    \end{align}
%\end{corollary}
%\begin{proof}
%    Re-organising $z_3^2 \geq 2z_3z_4$ for $\hat{R}_\mathrm{max}$ yields the lower bound on $\hat{R}_\mathrm{max}$.
%\end{proof}

%\textit{Remark 1:} Note that the lower bound on $\hat{R}_\mathrm{max}$ presented in Corollary \ref{corollary: C1} is only attainable if the required $\tau_\mathrm{min}$ does not exceed $\tau_\mathrm{max}$, i.e., $\tau_\mathrm{min} > \tau_\mathrm{max}$.

\textit{Remark:} Note that the pulse duration $\tau \in \mathcal{T}$ impacts the feasibility region of $A$ and $\boldsymbol{w}$ as these optimisation variables are coupled in Problem \eqref{P1: Original Problem}.
%Specifically, after reorganising constraints \eqref{con1: average transmit power} and \eqref{con1: average transmit power}, $1/\tau$ is involved both in an upper bound on $A$ in \eqref{con1: average transmit power} as well as in a lower bound on $A$ in \eqref{con1: radar range accuary}, thereby causing a
The non-trivial relationship of the optimisation variables is investigated in Section \ref{Section: Results and Performance Evaluation}.
%Moreover, the objective \eqref{Objective1: Sum harvested power} increases linearly with $\tau$.  

\vspace*{-0.25cm}
\subsection{Problem Reformulation and Proposed Solution}\label{Section: Solution}
Obtaining a solution to Problem \eqref{P1: Original Problem} is challenging due to the non-convexity of the problem and the coupling of optimisation variables.
Our proposed approach for solving Problem \eqref{P1: Original Problem} is based on a one-dimensional grid search over $\tau$ and the application of SDR and SCA at every point of the grid.
%For obtaining a solution to Problem \eqref{P1: Original Problem}, we perform a one-dimensional grid search over $\tau$.
To perform a one-dimensional grid search over $\tau$, the interval $\mathcal{T}$ is discretised into a set of $n_{\tau}$ equally-spaced values and we denote this set by $\mathcal{T}_{n_{\tau}}$.

First, we decouple optimisation variables $A$ and $\boldsymbol{w}$ from $\tau$ and solve Problem \eqref{P1: Original Problem} for $A_\tau$ and $\boldsymbol{w}_\tau$ at every point on the grid, i.e., $\tau \in \mathcal{T}_{n_{\tau}}$.
In the following, we propose an approach for obtaining the solution $A^*_\tau$ and $\boldsymbol{w}^*_\tau$, which maximises the objective function \eqref{Objective1: Sum harvested power}, for a given pulse duration $\tau \in \mathcal{T}_{n_{\tau}}$.
%For every point on the grid, we propose the following for solving Problem \eqref{P1: Original Problem}.
%is based on a change of variables, the application of the semidefinite relaxation (SDR) technique. 
%The relaxed problem is solved through Sequential Convex Programming (SCP), whereby the non-convex objective function is approximated by it's first order Taylor polynomial.
%The pulse duration which maximises the objective with its corresponding $A^*$ and $\boldsymbol{w}^*$, represents the optimal pulse duration $\tau^*$.
%The proposed method is outlined in Algorithm \ref{alg: Approach}.
To this end, we apply the change of variables $\boldsymbol{v} = A_\tau \boldsymbol{w}_\tau$.
Next, we define matrix variable $\boldsymbol{V} = \boldsymbol{v} \boldsymbol{v}^H$ which, by construction, is Hermitian and has a rank of $1$, i.e., $\mathrm{rank}(\boldsymbol{V})=1$.
Note that $\boldsymbol{V}$ is a positive semi-definite (PSD) matrix, i.e., $\boldsymbol{V} \succcurlyeq \boldsymbol{O}$, where $\boldsymbol{O} \in \mathbb{R}^{N_\mathrm{t} \times N_\mathrm{t}}$ denotes the all-zero matrix.
This yields the following equivalent reformulation of Problem \eqref{P1: Original Problem} for a fixed value of $\tau \in \mathcal{T}_{n_{\tau}}$
\begin{subequations}\label{P2: Reformulated Problem}
\begin{alignat}{2}
&\underset{\boldsymbol{V} \succcurlyeq \boldsymbol{O}}{\text{maximise}} & \qquad & \! \Phi(\boldsymbol{V}) \label{Objective2: Sum harvested power}\\
& \text{subject to} & & \widehat{\text{C1}}\text{: }\varepsilon_1(\tau) \leq \mathrm{Tr}\left\{\boldsymbol{U} \boldsymbol{V} \right\}, \label{con2: radar range accuary}\\
& & & \widehat{\text{C2/3}}\text{: } \mathrm{Tr}\left\{
\boldsymbol{V} \right\} \leq \varepsilon_2(\tau), \label{con2: combined con} \\
%&&& \text{C2: } \mathrm{Tr}\left\{ \boldsymbol{V} \right\} \leq \frac{T(\tau)}{\tau} P_\mathrm{avg}, \label{con2: average transmit power} \\
%&&& \text{C3: } \mathrm{Tr}\left\{ \boldsymbol{V} \right\} \leq P_\mathrm{p}, \label{con2: peak transmit power} \\
& & & \widehat{\text{C4}}\text{: } \mathrm{Tr}\left\{\boldsymbol{H}_m \boldsymbol{V} \right\} \leq P_\mathrm{max}, \forall m=1\dots M, \label{con2: peak input power} \\
& & & \widehat{\text{C6}}\text{: } \mathrm{rank}(\boldsymbol{V}) = 1, \label{con2: rank one}
\end{alignat}
\end{subequations}
where $\Phi(\boldsymbol{V}) = (\tau / (T(\tau)) \sum_{m=1}^M \varphi(\mathrm{Tr}\left\{\boldsymbol{H}_m \boldsymbol{V} \right\})$ with $\boldsymbol{H}_m = \boldsymbol{h}_m \boldsymbol{h}_m^H$, $\forall m=1\dots M$, $\boldsymbol{U} = \boldsymbol{u} \boldsymbol{u}^H$, $\varepsilon_1(\tau) = (z^2 \, [T(\tau)]^2)/(\tau \,\hat{R}_\mathrm{max}^2) > 0$, and $\varepsilon_2(\tau) = \text{min}\{ (T(\tau)/\tau) \, P_\mathrm{avg}, P_\mathrm{p}\} > 0$.
Note that Problem \eqref{P2: Reformulated Problem} is non-convex due to the objective function \eqref{Objective2: Sum harvested power} and $\widehat{\text{C6}}$. 
We propose the application of SCA to solve Problem \eqref{P2: Reformulated Problem} \cite{Sun17}.
Furthermore, $\widehat{\text{C6}}$ is dropped.
In Proposition \ref{Proposition: Rank One Optimality}, we will show that the solution in every iteration of the SCA algorithm satisfies $\widehat{\text{C6}}$ implicitly.
%We apply SCA to the original problem by constructing a lower bound of the objective in each iteration of the algorithm and dropping the rank-one constraint, which yields a convex optimisation problem. 
%Note that the solution in each iteration satisfies the rank-one constraint implicitly, which is proven in Proposition 2.
%Next, Problem \eqref{P2: Reformulated Problem} is relaxed by dropping constraint $\widehat{\text{C6}}$:
%\begin{equation}\label{eq: P2-relaxed}
%    \underset{\boldsymbol{V} \succcurlyeq 0}{\text{maximise}} \,\,\,
%    \Phi(\boldsymbol{V})\quad
%    \text{subject to} \quad \widehat{\text{C1}}, \widehat{\text{C23}}, \widehat{\text{C4}}.
%\end{equation}
%The non-convexity of the objective function $\Phi(\boldsymbol{V})$ causes Problem \eqref{eq: P2-relaxed} to be non-convex.
%Therefore, we propose the application of SCA to solve Problem \eqref{eq: P2-relaxed} \cite{Sun17}.
In iteration $i\geq 0$ of the SCA algorithm, we construct the following lower bound of the objective function
\begin{align}\label{eq: underestimate objective}
    \Phi(\boldsymbol{V}) \geq \hat{\Phi}(\boldsymbol{V}, \boldsymbol{V}^{i}),
\end{align}
with 
% \begin{align*}
% \hat{\Phi}(\boldsymbol{V}, \boldsymbol{V}^{i}) = \Phi(\boldsymbol{V}^{i}) + \frac{\tau}{T(\tau)} \left(\sum_{m=1}^M \varphi^{\prime} \left(\mathrm{Tr}\left\{\boldsymbol{H}_m \boldsymbol{V}^{i} \right\} \right)
% \mathrm{Tr}\left\{\boldsymbol{H}_m\boldsymbol{V} - \boldsymbol{H}_m\boldsymbol{V}^{i} \right\} \right).
% \end{align*}
%is given by
\begin{align}\label{eq: underestimate objective 2}
\resizebox{\hsize}{!}{%
   $\hat{\Phi}(\boldsymbol{V}, \boldsymbol{V}^{i}) = \Phi(\boldsymbol{V}^{i}) + 
   \frac{\tau}{T(\tau)}
   \left(\sum_{m=1}^M \varphi^{\prime} \left(\mathrm{Tr}\left\{\boldsymbol{H}_m \boldsymbol{V}^{i} \right\} \right) \mathrm{Tr}\left\{\boldsymbol{H}_m\boldsymbol{V} - \boldsymbol{H}_m\boldsymbol{V}^{i} \right\} \right)\mathrm{.}$}
\end{align}
Here, $\varphi^{\prime} \left(\cdot \right)$ denotes the derivative of $\varphi(\cdot)$ with respect to the input power evaluated at $\mathrm{Tr}\left\{\boldsymbol{H}_m \boldsymbol{V}^i \right\}$ and $\boldsymbol{V}^{i}$ is the solution obtained in the $i$-th iteration of the algorithm. 
We utilise $\boldsymbol{V}^{0} = (P_\mathrm{p}/\| \boldsymbol{u} \|_2^2) \left(\boldsymbol{u}\boldsymbol{u}^H\right)$ as a feasible initialisation of the algorithm to satisfy the desired sensing accuracy.
Consequently, the optimisation problem solved in every iteration $i$ of the algorithm is given by
\begin{equation}\label{eq: P2-relaxed-SCP}
    \boldsymbol{V}^{i+1} = \underset{\boldsymbol{V} \succcurlyeq 0}{\text{argmax}} \,\,\,
    \hat{\Phi}(\boldsymbol{V}, \boldsymbol{V}^{i})\quad
    \text{subject to} \quad \widehat{\text{C1}}, \widehat{\text{C2/3}}, \widehat{\text{C4}},
\end{equation}
which is a convex optimisation problem that can be solved efficiently with, for example, CVXPY \cite{diamond2016cvxpy}.
Next, we show that the solution of Problem \eqref{eq: P2-relaxed-SCP} yields a rank-one matrix.

%Next, we show that solving Problem \eqref{eq: P2-relaxed} is sufficient for obtaining a solution for Problem \eqref{P2: Reformulated Problem}.

\begin{proposition}\label{Proposition: Rank One Optimality}
    In each iteration $i\geq 0$, the solution $\boldsymbol{V}^{i+1}$ of Problem \eqref{eq: P2-relaxed-SCP} satisfies $\mathrm{rank}(\boldsymbol{V}^{i+1}) = 1$. 
    %Thus, by satisfying \eqref{con2: rank one} implicitly, $\Tilde{\boldsymbol{V}}^*$ is equivalent to the solution of Problem \eqref{P2: Reformulated Problem} denoted $\boldsymbol{V}^*$, i.e., $\Tilde{\boldsymbol{V}}^* = \boldsymbol{V}^*$. 
\end{proposition}
\begin{proof}
    The proof is provided in Appendix \ref{appendix B}.
\end{proof}

%The proof is similar to the proof of Proposition \ref{Proposition: Rank One Optimality}. 
The beamforming vector $\boldsymbol{w}_{\tau}^*$ and the signal amplitude $A_{\tau}^*$ are obtained as the dominant normalised eigenvector and the square root of the corresponding eigenvalue of the rank-one matrix $\boldsymbol{V}^{i}$ obtained in the final iteration of the algorithm, respectively.
%Consequently, $\boldsymbol{w}^*_\tau = \boldsymbol{v}^* / \| \boldsymbol{v}^* \|_2$ and $A^*_\tau$ is the square root of the eigenvalue associated to $\boldsymbol{v}^*$.
The proposed algorithm is summarised in Algorithm \ref{alg: Approach}. 
Note that the proposed algorithm converges to a stationary point of \eqref{P1: Original Problem} for $\mathcal{T}_{n_{\tau}}$ \cite{Lanckriet09}. 
The computational complexity of a single iteration of the algorithm is $\mathcal{O}\left(M N_\mathrm{t}^{3.5} + M^2 N_\mathrm{t}^{2.5} + M^3 N_\mathrm{t}^{0.5} \right)$, where $\mathcal{O}(\cdot)$ denotes the big-O notation \cite{Shanin22}.
Lastly, the optimal pulse duration is found by evaluating $\tau^* ={\text{argmax}}_{\substack{\tau}} \,\,\, \phi(A^*_\tau \boldsymbol{w}^*_\tau, \tau)$.

\begin{algorithm}
\caption{Procedure for obtaining $\tau^*, A^*, \boldsymbol{w}^*$}\label{alg: Approach}
\begin{algorithmic}
\State\textbf{Initialise: }{$\boldsymbol{V}^{0} = \frac{P_\mathrm{p}}{\| \boldsymbol{u} \|_2^2} \left(\boldsymbol{u}\boldsymbol{u}^H\right)$, $i=0$, $h^{0}=0$, $h^{-1}=2\, \epsilon_\mathrm{SCA}$ with $\epsilon_\mathrm{SCA}$ denoting the tolerance.}
\For{$\tau \in \mathcal{T}_{n_{\tau}}$}
\While{$\vert h^{i} - h^{i-1} \vert > \epsilon_\mathrm{SCA}$}
\State{Obtain $\boldsymbol{V}^{i+1}$ by solving \eqref{eq: P2-relaxed-SCP} for $\boldsymbol{V}^{i}$}
\State{Determine $h^{i+1} = \Phi(\boldsymbol{V}^{i+1})$ and set $i=i+1$}
\EndWhile
\State{$\boldsymbol{w}_{\tau}^*$ is the dominant normalised eigenvector of $\boldsymbol{V}^{i}$.}
%Extract dominant eigenvector $\boldsymbol{v}^*$ from $\boldsymbol{V}^{i}$.} 
%\State{Determine $\boldsymbol{w}^*_\tau = \boldsymbol{v}^* / \| \boldsymbol{v}^* \|_2$.}
\State{$A^*_\tau$ is the square root of the eigenvalue associated \\\textcolor{white}{......}with $\boldsymbol{w}_{\tau}^*$.}
\EndFor
\State{Determine $\tau^* = {\text{argmax}}_{\substack{\tau}} \,\,\, \phi(A^*_\tau \boldsymbol{w}^*_\tau, \tau)$}
\State\textbf{Output: }{$\tau^*, A^* = A^*_{\tau^*}, \boldsymbol{w}^* = \boldsymbol{w}^*_{\tau^*}$}
\end{algorithmic}
\end{algorithm}

\section{Results and Performance Evaluation}
\label{Section: Results and Performance Evaluation}
In the following, we consider an ISAPT system with $M=3$ EH nodes with weights $\beta_m = 1/M$, $\forall m=1,\dots,M$, which are located at a distance of $5$ m from the TRX at $45\degree$, $60\degree$, and $75\degree$, respectively. 
The channel $\boldsymbol{h}_m$, $\forall m=1,\dots,M$, between the TRX and EH node $m$ is modelled as a Rician fading channel.  
Thereby, we calculate the respective path loss as $\lambda^2 / (4 \pi R_m)^2$, where $R_m$ is the distance to the $m$-th EH node, and assume Rician K-factor $\kappa=1$.
Moreover, the ST is located at $\alpha = -60\degree$.
We assume $\sigma_\mathrm{RCS}=1$ $\text{m}^2$, which is a typical value for several applications and, e.g., approximately the RCS of a human that may be equipped with a wearable device \cite{Skolnik81}. 
%We consider scenario S-D, where the ST is located at $D_\mathrm{max} = 10$ m and $-60\degree$ from the TX. 
%Additionally, we investigate a scenario S-C in Section \ref{Section: optimal pulse duration}, where the ST is co-located with the EH located at $5$ m and $60\degree$ from the TX.
Table \ref{tab: simulation parameters} provides all relevant simulation parameters.
All results are averaged over $100$ EH channel realisations.
\vspace*{-0.15cm}
\setlength{\tabcolsep}{2pt}
\begin{table}[h]
\caption{Simulation parameters.}
\label{tab: simulation parameters}
\centering
\footnotesize   
\begin{tabular}{l|l|l}
General parameters & WPT parameters & Sensing parameters  \\ 
\hline
$P_\mathrm{avg} \in \{0.1,0.5\}$ W  & $a = 1.29$  & $B=10$ MHz \\
$P_\mathrm{p} \in \{0.5,1\}$ W & $C = 1.55 \cdot 10^3$ & $\sigma_\mathrm{n}^2 = -80$ dBm \\
$\lambda = 0.125$ m & $I_s = 5$ µA & $R_\mathrm{max} = 20$ m \\
Channel realisations: $100$ & $R_L = 10$ k$\Omega$ & $R_\mathrm{min} \in \{3,5,18\}$ m \\
$N_\mathrm{t} = 10$ & $P_\mathrm{max} = 25$ µW & $\hat{R}_\mathrm{max} \in [0.01, 0.06]$ m \\
$\epsilon_\mathrm{SCA} = 1 \cdot 10^{-7}$ & $M=3$ & $\sigma_\mathrm{RCS}$ = $1$ $\text{m}^2$ \\
$n_{\tau}=50$ & Rician K-factor: $1$ & $T_\mathrm{sen} = 1$ ms \\
$\Delta_\mathrm{TRX} = \lambda/2$ & $T_\mathrm{coh} = 1$ ms &
\end{tabular}
\vspace*{-0.5cm}
\end{table}

%\vspace*{-1cm}
\subsection{Optimal Pulse Duration}\label{Section: optimal pulse duration}
First, we investigate the average amount of harvested power \eqref{eq: sum harvested power} for all $\tau \in \mathcal{T}_{n_{\tau}}$, i.e., from the smallest to the largest value of $\mathcal{T}$, thereby determining the optimal pulse duration $\tau^*$.
%$ = \underset{\tau}{\text{argmax}} \,\,\, \phi(A^*_\tau \boldsymbol{w}^*_\tau)$.
To this end, we set $\hat{R}_\mathrm{max} = 0.02$ m, $R_\mathrm{min} = 18$ m, and $P_\mathrm{p} = 0.5$ W and investigate the average amount of harvested power \eqref{eq: sum harvested power} versus $\tau \in \mathcal{T}_{n_{\tau}}$ for $P_\mathrm{avg} = 0.1$ W and $P_\mathrm{avg} = 0.5$ W, respectively, in Fig. \ref{fig: Optimal_Tau}.
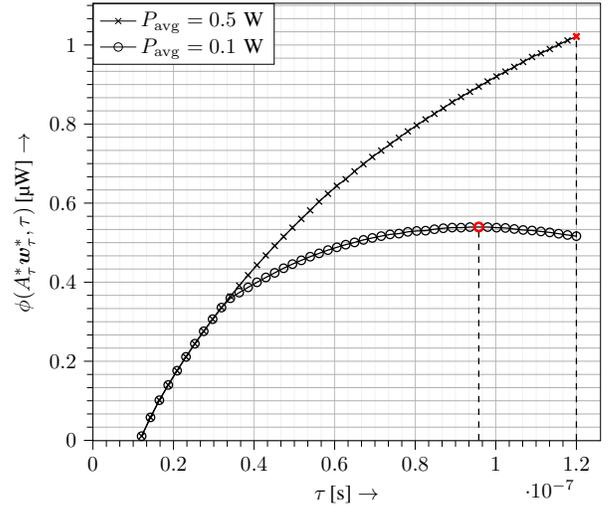
\begin{figure}[t]
    \centering
    \scalebox{0.8}{
    \definecolor{black}{RGB}{0,0,0}
\definecolor{bluish-green}{RGB}{0,158,115}
\definecolor{darkgray176}{RGB}{176,176,176}

\begin{tikzpicture}
\begin{axis}[
scale only axis,
%axis y line*=left,
tick align=outside,
tick pos=left,
unbounded coords=jump,
xlabel={$\tau \,\text{[s]} \rightarrow$},
xmin=0, xmax=1.259899051726205e-07,
xtick style={color=black},
ylabel={$\phi(A^*_\tau \boldsymbol{w}^*_\tau, \tau) \,\text{[µW]} \rightarrow$},
x tick scale label style={yshift=4pt},
y grid style={darkgray176},
ymin=-0.001,
ymax=1.10548121836406,
ytick style={color=black},
grid=both,
grid style={line width=.1pt, draw=gray!10},
major grid style={line width=.2pt,draw=gray!50},
minor tick num=5,
legend cell align=left,
legend pos=north west,
legend style={
    at={(0,1)},
    legend columns=1,
    legend entries={$P_\mathrm{avg}=0.5$ W, $P_\mathrm{avg}=0.1$ W}
    }
]
\addlegendimage{mark=x,black};
\addlegendimage{mark=o,black};

\addplot [mark=x,
semithick,
black]
table{%
x  y
1.20854589233017e-08 0.0105543736648982
1.42877964962956e-08 0.0579723754697767
1.64901340692894e-08 0.10148151861251
1.86924716422833e-08 0.139974390211485
2.08948092152771e-08 0.176156056942083
2.30971467882709e-08 0.21108494103143
2.52994843612648e-08 0.244407187342341
2.75018219342586e-08 0.275848578287492
2.97041595072525e-08 0.306410115714261
3.19064970802463e-08 0.335796790255513
3.41088346532402e-08 0.364069810974975
3.6311172226234e-08 0.391605666800934
3.85135097992278e-08 0.417756510695194
4.07158473722217e-08 0.443178497083286
4.29181849452155e-08 0.467697224721319
4.51205225182094e-08 0.491951098604464
4.73228600912032e-08 0.515102406739282
4.9525197664197e-08 0.537858351149504
5.17275352371909e-08 0.560042053695387
5.39298728101847e-08 0.581693272403086
5.61322103831786e-08 0.603949907110143
5.83345479561724e-08 0.623858502083382
6.05368855291663e-08 0.643789026169484
6.27392231021601e-08 0.660430600703136
6.49415606751539e-08 0.680374038976849
6.71438982481478e-08 0.69864359956389
6.93462358211416e-08 0.716676695293457
7.15485733941355e-08 0.732845750871136
7.37509109671293e-08 0.748978199964812
7.59532485401232e-08 0.765320791060556
7.8155586113117e-08 0.78151269094099
8.03579236861108e-08 0.796550453875794
8.25602612591047e-08 0.811706129007339
8.47625988320985e-08 0.825933208358391
8.69649364050924e-08 0.839994155156295
8.91672739780862e-08 0.855404707496351
9.136961155108e-08 0.868333115299961
9.35719491240739e-08 0.881646282914094
9.57742866970677e-08 0.894887507213601
9.79766242700616e-08 0.907755184948181
1.00178961843055e-07 0.920398102105253
1.02381299416049e-07 0.932608681066601
1.04583636989043e-07 0.944593366061215
1.06785974562037e-07 0.957014751340991
1.08988312135031e-07 0.969682092531286
1.11190649708025e-07 0.978981292570737
1.13392987281018e-07 0.989892558166551
1.15595324854012e-07 1.00094175033605
1.17797662427006e-07 1.01147685824085
1.2e-07 1.02150551956007
};

\addplot [mark=o,
semithick,
black]
table{%
x  y
1.20854589233017e-08 0.0105543736648982
1.42877964962956e-08 0.0579723754697767
1.64901340692894e-08 0.10148151861251
1.86924716422833e-08 0.139974390211485
2.08948092152771e-08 0.176156056942083
2.30971467882709e-08 0.21108494103143
2.52994843612648e-08 0.244407187342341
2.75018219342586e-08 0.275848578287492
2.97041595072525e-08 0.306410115714261
3.19064970802463e-08 0.335796790255513
3.41088346532402e-08 0.359464621227767
3.6311172226234e-08 0.37360659107542
3.85135097992278e-08 0.386907458237972
4.07158473722217e-08 0.39973265602769
4.29181849452155e-08 0.41201227854956
4.51205225182094e-08 0.423570052057753
4.73228600912032e-08 0.435275712635746
4.9525197664197e-08 0.445709026415789
5.17275352371909e-08 0.455224004412808
5.39298728101847e-08 0.464525749126044
5.61322103831786e-08 0.4728654585385
5.83345479561724e-08 0.481169266524595
6.05368855291663e-08 0.4882255185374
6.27392231021601e-08 0.494914126152846
6.49415606751539e-08 0.500633088678763
6.71438982481478e-08 0.506982518505368
6.93462358211416e-08 0.511923705414028
7.15485733941355e-08 0.516558555464247
7.37509109671293e-08 0.520862122654858
7.59532485401232e-08 0.522853586131858
7.8155586113117e-08 0.527393919379342
8.03579236861108e-08 0.5295969176386
8.25602612591047e-08 0.530303161740121
8.47625988320985e-08 0.533921678706804
8.69649364050924e-08 0.53642957138311
8.91672739780862e-08 0.538115900457247
9.136961155108e-08 0.537978783318337
9.35719491240739e-08 0.539374009500866
9.57742866970677e-08 0.539512647686132
9.79766242700616e-08 0.53937924215572
1.00178961843055e-07 0.5379316006687
1.02381299416049e-07 0.537206426879762
1.04583636989043e-07 0.534542160060663
1.06785974562037e-07 0.532051689572896
1.08988312135031e-07 0.530797654630793
1.11190649708025e-07 0.52833344558908
1.13392987281018e-07 0.525905537760574
1.15595324854012e-07 0.522627536597195
1.17797662427006e-07 0.519446822814869
1.2e-07 0.516461556274798
};

\addplot [mark=o,
very thick,
red]
table{%
x  y
9.57742866970677e-08 0.539512647686132
};

\addplot [mark=x,
very thick,
red]
table{%
x  y
1.2e-07 1.02150551956007
};

\draw [dashed, semithick] (9.57742866970677e-08,0) -- (9.57742866970677e-08, 0.539512647686132);

\draw [dashed, semithick] (1.2e-07,0) -- (1.2e-07,1.02150551956007);

%\node[above, draw, circle] at (9.57742866970677e-08, 0.539512647686132) (nodeA) {$\tau^*$};

%\node[above] at (1.2e-07,1.02150551956007) (nodeB) {$\tau^*$};

\end{axis}

\end{tikzpicture}
    }
    \vspace*{-1mm}
    \caption{Average harvested power $\phi(A^*_\tau \boldsymbol{w}^*_\tau, \tau)$ for $\tau \in \mathcal{T}_{n_\tau}$. The optimal pulse duration $\tau^*$ is highlighted in red.}
    \vspace*{-2mm}
    \label{fig: Optimal_Tau}
\end{figure}
%Generally, at least one of the EH nodes is driven into saturation except when $\tau \in \mathcal{T}_{n_{\tau}}$ is very small. 
%In the latter case, none of the EH receivers can be driven into saturation since the majority of the transmit power has to be transmitted in the direction of the ST to achieve the desired sensing accuracy. 
Next, the non-trivial relationship among the optimisation variables is studied.
%For increasing $\tau$, the objective function \eqref{eq: sum harvested power} increases, satisfying constraint C1, i.e., the sensing accuracy constraint, becomes less challenging, however, constraint C2, i.e., the average power constraint, becomes increasingly hard to satisfy.
When $P_\mathrm{avg}$ is large, i.e., $P_\mathrm{avg} = 0.5$ W, the average harvested power increases with $\tau \in \mathcal{T}_{n_\tau}$, and thus, $\tau^* = \tau_\mathrm{max} = 1.2 \cdot 10^{-7}$ s, which is highlighted by the red cross in Fig. \ref{fig: Optimal_Tau}.
Note that $\tau^* / T(\tau^*) \approx 47\%$.
In fact, constraint C2 is not tight in this case and the solution of Problem \eqref{P1: Original Problem} is determined by the sensing accuracy constraint C1.
%When constraint C2 is not tight, the average harvested power \eqref{eq: sum harvested power} increases monotonically with $\tau$ and thus, $\tau^* = \tau_\mathrm{max} = 1.2 \cdot 10^{-7}$ s, which is highlighted by the red cross in Fig. \ref{fig: Optimal_Tau}.
However, this behaviour does not apply for low $P_\mathrm{avg}$, i.e., $P_\mathrm{avg} = 0.1$ W, as, in this case, constraint C2 becomes increasingly important.
Then, as $\tau$ grows, we first observe an increase of the harvested power up to $\tau^* \approx 0.96 \cdot 10^{-7}$ s, which is highlighted by the red circle in Fig. \ref{fig: Optimal_Tau}, followed by a decrease of harvested power.
In this case, $\tau^* / T(\tau^*) \approx 38\%$.
Consequently, determining the optimal pulse duration $\tau^*$ of the ISAPT system is important to maximise the harvested power at the EH nodes.

\vspace*{-0.1cm}
\subsection{Performance Evaluation and Trade-Off Between Sensing Accuracy and Harvested Power}
Next, the trade-off between sensing performance and WPT is investigated.
We compare the performance of our proposed solution to a baseline scheme, which linearly combines the optimal transmit signal strategy for WPT when assuming a linear EH model, i.e., energy beamforming \cite{Zhang13}, with the optimal radar sensing signal, i.e., the beamsteering vector towards the ST. 
%whereby an approximate location of the ST is assumed to be available from a previous time slot.
%Moreover, 
The weighting factor for the linear combination of the two signals is set to satisfy constraint C1 with equality, such that only the minimum required sensing accuracy is achieved and the remaining transmit power is dedicated to WPT.
We utilise the optimal pulse duration for the proposed scheme and the baseline scheme, which are both obtained through the grid search over $\tau$.

For the results in Fig. \ref{fig: Big_Result}, we set $P_\mathrm{avg}=0.5$ W and investigate the trade-off between the average harvested power \eqref{eq: sum harvested power} and the sensing accuracy $\hat{R}_\mathrm{max}$ for different $R_\mathrm{min}$ and $P_\mathrm{p}$. 
%and thus, $\tau^* = \tau_\mathrm{max}$ in the following.
\begin{figure}[t]
    \centering
    \scalebox{0.8}{
    \input{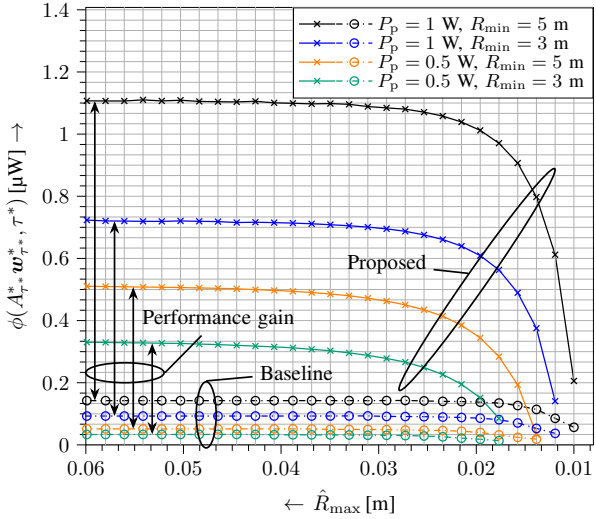}
    }
    \vspace*{-1mm}
    \caption{Average harvested power $\phi(A^*_{\tau^*} \boldsymbol{w}^*_{\tau^*}, \tau^*)$ vs. $\hat{R}_\mathrm{max}$. Solid lines with crosses correspond to the proposed solution, while dash-dotted lines with circles indicate the baseline scheme.}
    \vspace*{-2mm}
    \label{fig: Big_Result}
\end{figure}
We observe from Fig. \ref{fig: Big_Result} that the proposed solution (solid lines with crosses) outperforms the baseline scheme (dash-dotted lines with circles) by a considerable margin indicated by the arrows in Fig. \ref{fig: Big_Result}.
For larger $P_\mathrm{p}$, the ISAPT system can achieve higher sensing accuracy and average harvested power.
The same observation is made for larger $R_\mathrm{min}$ since a longer pulse duration $\tau$ is admissible for larger $R_\mathrm{min}$, which follows from \eqref{eq: minimum distance constraint}.
%This is particularly critical when an amount of power exceeding than the harvestable amount is transmitted to an EH receiver, the average sum of harvested power is affected negatively if this power could have been harvested by another EH node in the system. 
Furthermore, we observe a trade-off between average harvested power \eqref{eq: sum harvested power} and sensing accuracy.
In fact, for a smaller tolerated maximum range estimation error $\hat{R}_\mathrm{max}$, more power has to be transmitted towards the ST to attain a stronger echo signal, and thus, the average harvested power decreases.
Thus, the joint optimisation of the transmit signal and beamforming design taking into account the non-linearity of EH circuits, is essential for efficient ISAPT.
%by jointly optimising the transmit signal and the beamforming design to maximise the sum of average harvested powers, the proposed solution integrates the WPT and sensing functionalities more efficiently compared to the baseline scheme, which neglects the EH non-linearity.
%For given $\hat{R}_\mathrm{max}$, the loss in average harvested power to satisfy the desired sensing accuracy is found to be more severe when utilising $P_\mathrm{p}=0.5$ W instead to $P_\mathrm{p}=1$ W for a given value of $R_\mathrm{min}$.
%compared to a reduction of $R_\mathrm{min}$ from $5$ m to $3$ m for the same $P_\mathrm{p}$.
%Moreover, a larger $R_\mathrm{min}$, which implies a larger $\tau_\mathrm{max}$, yields more harvested power increases at the same $P_\mathrm{p}$.
%Note that $\tau^* = \tau_\mathrm{max}$ increases with increasing $R_\mathrm{min}$ thus, yielding more harvested power.
%Thereby, the minimum coverage distance constraint impacts the average harvested power more severely compared to the range accuracy constraint.

\vspace*{-0.1cm}
\section{Conclusion}\label{Section: Conclusion}
In this paper, we jointly optimised the rectangular transmit pulse and the beamforming vector for an ISAPT system with multiple EH nodes and a single ST while considering a non-linear circuit-based EH model.
Hereby, the objective was to maximise the weighted sum of the average harvested powers at the EH receivers
in a time slot while ensuring the required sensing performance for estimating the range to the ST with a desired accuracy.
We determined the feasible region of the pulse duration analytically and proposed a low-complexity solution to the non-convex problem based on a one-dimensional grid search, SDR, and SCA.
Analytically, the solution at each iteration of the SCA algorithm was shown to be a rank-one matrix, thus yielding the optimal beamformer and amplitude of the transmit signal, respectively.
The proposed low-complexity solution significantly outperforms a baseline scheme, the transmit strategy of which is based on a linear combination of energy beamforming and the radar beamsteering vector.
This underscores the importance of accurately modelling of the non-linear EH circuit for efficient ISAPT system design.
%, and showed that the sensing capability of the ISAPT system is limited by the amount of available transmit power and the minimum distance at which the ST may be located. 
%For given desired sensing accuracy, the loss in average harvested power to satisfy the desired sensing accuracy is found to be more severe when using a lower peak transmit power.
%constraint  minimum distance constraint impacts the average harvested power more severly than the range accuracy constraint.
A generalisation of the problem to multiple STs or the incorporation of additional communication links present interesting topics for future work.

\appendices

\section{}\label{appendix A}
    A lower bound on pulse duration $\tau$ is obtained by assuming that the tolerated accuracy level in C1 is achieved with equality and all the available power is utilised for ST localisation, i.e.,
    %re-organising \eqref{con1: radar range accuary} for $\tau$ and, subsequently, bounding the numerator and the denominator from below and from above, respectively.
    %The shortest pulse duration $\tau_\mathrm{min}$, only just manages to achieve the tolerated accuracy level, i.e., C1 holds with equality, when all the available power is utilised for ST localisation, i.e.,
    $A=\sqrt{P_\mathrm{p}}$ and $\boldsymbol{w} = \boldsymbol{u} / \| \boldsymbol{u} \|_2$.
    This leads to the following condition for the lower bound $\tau_\mathrm{min}$
    \begin{align}\label{eq: tau min equation}
        \frac{z^2 \, [T(\tau_\mathrm{min})]^2}{\tau_\mathrm{min} \,\hat{R}_\mathrm{max}^2} = P_\mathrm{p} \| \boldsymbol{u} \|_2^2\mathrm{.}
    \end{align}
    Rewriting \eqref{eq: tau min equation}, yields a quadratic equation for $\tau_\mathrm{min}$ as follows
    \begin{align}\label{eq: tau min equation 2}
        \tau_\mathrm{min}^2 + \tau_\mathrm{min} \left[ z_4 - z_3 \right] + \frac{4 \, R_\mathrm{max}^2}{c^2} = 0,
    \end{align}
    with $z_3 =(P_\mathrm{p} \| \boldsymbol{u} \|_2^2 \hat{R}_\mathrm{max}^2)/z^2 > 0$ and $z_4 = (4R_\mathrm{max})/c > 0$, which has the following two solutions
    \begin{align}\label{eq: tau min 1}
       \tau_\mathrm{min}^{(\text{I}),(\text{II})} = \frac{1}{2} \left( z_3 - z_4 \pm \sqrt{z_3^2 - 2z_3 z_4} \right)\mathrm{.}
    \end{align}
    The smallest feasible $\tau_\mathrm{min}$ is given by the minimum of $\tau_\mathrm{min}^{(\text{I})}$ and $\tau_\mathrm{min}^{(\text{II})}$, which has to be positive for any $z_3$ and $z_4$ to satisfy C1.
    In the following, we show that $\tau_\mathrm{min} = \frac{1}{2} \left( z_3 - z_4 - \sqrt{z_3^2 - 2z_3 z_4} \right) > 0$ holds.
    Since $\tau_\mathrm{min}$ is real-valued, $z_3^2 \geq 2z_3z_4$ must hold in \eqref{eq: tau min 1}, and this in turn implies $z_3 > z_3 /2 \geq z_4 > 0$, i.e., $z_3-z_4 >0$.
    Moreover,   
    \begin{align}
    \resizebox{\hsize}{!}{%
        $z_3 - z_4 > \sqrt{z_3^2 - 2z_3 z_4}
        \iff z_3^2 - 2z_3z_4 + z_4^2 > z_3^2 - 2z_3 z_4\mathrm{,}$}
    \end{align}
    holds by definition since $z_4>0$, which concludes the proof.
    %. and we solve for $\tau_\mathrm{min}$. Beyond this, since the equation of obtaining $\tau_\mathrm{min}$ is now quadratic, we can define an additional system constraint to define the minimum possible sensing accuracy based on the system parameters. This constraint stems from ensuring that $\tau_\mathrm{min}$ is a real number when investigating the roots of the equation. Consequently, this is a justifiable starting point for the simulations. 
    %The numerator is minimised by utilising \eqref{eq: minimum time frame duration}, i.e., by setting $T=T_\mathrm{min}$. 
    %The denominator is bounded from above for $\boldsymbol{w} = \boldsymbol{g} / \| \boldsymbol{g} \|_2$, i.e., if the beamforming vector $\boldsymbol{w}$ is matched to the channel to the ST $\boldsymbol{g}$, and the maximum amount of power is utilised by setting $A=\sqrt{P_\mathrm{p}}$ to satisfy the range accuracy constraint.
    %Thus, the lower bound on the pulse duration $\tau$ is given by
    %\begin{align}\label{eq: tau_min}
    %    \frac{c \, \sigma_\mathrm{n}^2 \, 2 \, R_\mathrm{max}}{16 \, B^2 \, \hat{R}_\mathrm{max}^2 \, P_\mathrm{p} \, \| \boldsymbol{g} \|_2^4} \leq \tau.
    %\end{align}
    %Since an upper bound on $\tau$ is imposed in \eqref{eq: tau_max}, the feasibility region $\mathcal{T}$ of $\tau$ is given by \eqref{eq: tau_max} and \eqref{eq: tau_min}.

\vspace*{-0.1cm}
\section{}\label{appendix B}
\vspace*{-0.1cm}
%Problem \eqref{eq: P2-relaxed-SCP} is equivalently reformulated as follows:
    %A solution $\Tilde{\boldsymbol{V}}^*$ of the Problem \eqref{eq: P2-relaxed-SCP} optimally solves the following convex optimisation problem
    %\begin{subequations}\label{P3: Problem with no objective}
    %\begin{alignat}{2}
    %&\underset{\boldsymbol{V} \succcurlyeq 0}{\text{maximise}}
    %&\qquad& \! \hat{\Phi}(\boldsymbol{V}, \boldsymbol{V}^{i})
    %\label{Objective3: Sum harvested power}\\
    %&\text{subject to} 
    %&& \text{C1: } \varepsilon_1 \leq \mathrm{Tr}\left\{\boldsymbol{G} %\boldsymbol{V} \right\}, \label{con3: radar range accuary}\\
    %&&& \text{C2: } \mathrm{Tr}\left\{ \boldsymbol{V} \right\} \leq %\varepsilon_2, \label{con3: power limit} \\
    %&&& \text{C3: } \mathrm{Tr}\left\{\boldsymbol{H}_m \boldsymbol{V} %\right\} \leq P_\mathrm{max}, \forall m=1\dots M, \label{con3: peak input power},
    %\end{alignat}
    %\end{subequations}
    %where $\varepsilon_1 = \left[\left(2 \, c R_\mathrm{max} + c^2 \tau \right) \, \sigma_\mathrm{n}^2)/(16 \, B \tau \hat{R}^2 \, \mathrm{Tr}\left\{\boldsymbol{G}\right\})\right] > 0$ and $\varepsilon_2 = \text{min}\left\{((2 \, R_\mathrm{max})/c + \tau)/\tau P_\mathrm{avg}, P_\mathrm{p} \right\}$.
    The gap between Problem \eqref{eq: P2-relaxed-SCP} and its dual problem is equal to zero since strong duality holds \cite{Boyd04}.
    The Lagrangian of Problem \eqref{eq: P2-relaxed-SCP} is given by
    \begin{align}\label{eq: Lagrangian P3}
        \mathcal{L} = &- \hat{\Phi}(\boldsymbol{V}, \boldsymbol{V}^{i}) -\mu \mathrm{Tr}\left\{\boldsymbol{U} \boldsymbol{V} \right\} + \xi \mathrm{Tr}\left\{\boldsymbol{V} \right\} \nonumber \\ &+ \sum_{m=1}^M \mu_m \mathrm{Tr}\left\{\boldsymbol{H}_m \boldsymbol{V} \right\} - \mathrm{Tr}\left\{\boldsymbol{Y} \boldsymbol{V} \right\} + \gamma,
    \end{align}
    where $\mu$, $\xi$, and $\mu_m$, $\forall m=1,\dots,M$, are the Lagrangian multipliers related to constraints $\widehat{\text{C1}}$, $\widehat{\text{C2/3}}$, and $\widehat{\text{C4}}$, respectively.
    Moreover, $\boldsymbol{Y}$ is the Lagrangian multiplier associated with the constraint restricting $\boldsymbol{V}$ to a PSD matrix and $\gamma$ accounts for all terms not involving $\boldsymbol{V}$.
    We note that the Karush-Kuhn-Tucker (KKT) conditions are satisfied for the optimal solution $\boldsymbol{V}^{i+1}$ of Problem \eqref{eq: P2-relaxed-SCP} and $\check{\mu}, \check{\xi}, \check{\mu}_m$, $\forall m=1,\dots,M$, and $\check{\boldsymbol{Y}}$ are the solutions of the dual problem of Problem \eqref{eq: P2-relaxed-SCP}.
    The KKT conditions are given by
    \begin{subequations}\label{eq: KKT}
    \begin{align}
        &\nabla_{\boldsymbol{V}} \mathcal{L} = \boldsymbol{O} \label{eq: Lagrangian gradient}\\
        &\check{\mu} \geq 0, \check{\xi} \geq 0, \check{\mu}_m \geq 0, \forall m=1,\dots,M, \check{\boldsymbol{Y}} \succcurlyeq 0 \label{eq: non-neg con} \\
        &\check{\boldsymbol{Y}} \check{\boldsymbol{V}} = \boldsymbol{O} \label{eq: comp slack},
        \end{align}
    \end{subequations}
    where $\nabla_{\boldsymbol{V}} \mathcal{L}$ is the gradient of $\mathcal{L}$ with respect to $\boldsymbol{V}$ and the matrix $\boldsymbol{O} \in \mathbb{R}^{N_\mathrm{t} \times N_\mathrm{t}}$ denotes the all-zero matrix.
    From \eqref{eq: Lagrangian gradient}, we obtain $\check{\boldsymbol{Y}} = \check{\xi} \boldsymbol{I} - \boldsymbol{Z}$,
    where 
    %$\chi =  \check{\xi} -\left( \tau / T(\tau) \right) \left(\sum_{m=1}^M \varphi^{\prime} \left(\mathrm{Tr}\left\{\boldsymbol{H}_m \boldsymbol{V}^{i} \right\} \right) \boldsymbol{H}_m \right) $,
    $\boldsymbol{I} \in \mathbb{R}^{N_\mathrm{t} \times N_\mathrm{t}}$ is the identity matrix of size $N_\mathrm{t}$, and $\boldsymbol{Z} = \check{\mu} \boldsymbol{U} - \sum_{m=1}^M \left(\check{\mu}_m - (\tau / T(\tau)) \varphi^{\prime} \left(\mathrm{Tr}\left\{\boldsymbol{H}_m \boldsymbol{V}^{i} \right\} \right)\right) \boldsymbol{H}_m$.
    Note that by definition of $\boldsymbol{U}$ and $\boldsymbol{H}_m$, $\forall m=1, \dots, M$, $\boldsymbol{Z}$ is a Hermitian matrix, i.e., $\boldsymbol{Z}^H = \boldsymbol{Z}$, and thus, it can be expressed as $\boldsymbol{Z} = \boldsymbol{P} \boldsymbol{\Lambda} \boldsymbol{P}^H$, where the columns of unitary matrix $\boldsymbol{P}$ are the eigenvectors of $\boldsymbol{Z}$ and $\boldsymbol{\Lambda}$ is a diagonal matrix containing the corresponding real-valued eigenvalues.
    Consequently, $\check{\boldsymbol{Y}}$ is given by
    \begin{align}\label{eq: Lagrangian grad}
        \check{\boldsymbol{Y}} = \boldsymbol{P} \left( \check{\xi} \boldsymbol{I} - \boldsymbol{\Lambda} \right) \boldsymbol{P}^H.
    \end{align}
    
    Let $\delta_\mathrm{max}$ denote the largest eigenvalue of $\boldsymbol{Z}$, which, with probability $1$, has algebraic multiplicity $1$ due to the randomness of the channel.
    If $\check{\xi} > \delta_\mathrm{max}$, then $\check{\boldsymbol{Y}}$ is invertible and thus, $\mathrm{rank}(\check{\boldsymbol{Y}}) = N_\mathrm{t}$.
    Consequently, $\boldsymbol{V}^{i+1} = \boldsymbol{O}$ with $\mathrm{rank}(\boldsymbol{V}^{i+1}) = 0$, which follows from \eqref{eq: comp slack}.
    However, $\boldsymbol{V}^{i+1} = \boldsymbol{O}$ is infeasible since it violates constraint $\widehat{\text{C1}}$.
    If $\check{\xi} < \delta_\mathrm{max}$, then at least one eigenvalue of $\check{\boldsymbol{Y}}$ is negative, which implies that $\check{\boldsymbol{Y}}$ is not a PSD matrix, thereby contradicting \eqref{eq: non-neg con}.
    Consequently, $\check{\xi} = \delta_\mathrm{max} \geq 0$ must hold, which implies $\mathrm{rank}(\check{\boldsymbol{Y}}) = N_\mathrm{t}-1$.
    The application of Sylvester's rank inequality to \eqref{eq: comp slack} yields
    \begin{align}\label{eq: rank proof}
        0 = \mathrm{rank}(\check{\boldsymbol{Y}} \boldsymbol{V}^{i+1}) \geq \mathrm{rank}(\check{\boldsymbol{Y}}) + \mathrm{rank}(\boldsymbol{V}^{i+1}) - N_\mathrm{t},
    \end{align}
    which implies $1 \geq \mathrm{rank}(\boldsymbol{V}^{i+1})$. 
    Consequently, the optimal solution $\boldsymbol{V}^{i+1}$ to Problem \eqref{eq: P2-relaxed-SCP} satisfies $\mathrm{rank}(\boldsymbol{V}^{i+1}) = 1$.
    
    %Due to the equivalence of problems \eqref{eq: P2-relaxed-SCP} and \eqref{P3: Problem with no objective}, $\check{\boldsymbol{V}} = \boldsymbol{V}^{i+1}$ must hold and thus, $\mathrm{rank}(\boldsymbol{V}^{i+1}) = 1$ applies.
    %Therefore, a solution to Problem \eqref{eq: P2-relaxed} equivalently solves Problem \eqref{P2: Reformulated Problem}, i.e., $\Tilde{\boldsymbol{V}}^* = \boldsymbol{V}^*$.

%\newpage
%\bibliographystyle{ieeetr}
%\bibliography{refs}
\printbibliography

\end{document}